\documentclass[10pt, conference]{IEEEtran}
\IEEEoverridecommandlockouts
\usepackage[english]{babel}
\usepackage{blindtext}
\usepackage{algorithm}
\usepackage{algorithmic}
\usepackage{graphicx}
\usepackage{subfigure}
\usepackage{amsmath}
\usepackage{amssymb}
\usepackage{cases}
\usepackage{setspace}
\usepackage{multirow}
\newtheorem{theorem}{\textbf{Theorem}}
\newtheorem{definition}{\textbf{Definition}}

\begin{document}
\title{Estimating the Distance Between Macro Base Station and Users in Heterogeneous Networks}
\author{\authorblockN{Lin Zhang$^{*\dag}$, Wenli Zhou$^{*}$, Wanbin Tang$^{*}$, Gang Wu$^{*}$, and Zhi Chen$^{*}$ \\ $*$ NCL, University of Electronic Science and Technology of China, Chengdu, China\\ $\dag$ Science and Technology on Information Transmission and Dissemination \\ in Communication Networks Lab., China
}
%\authorblockN{Wenli Zhou, Lin Zhang, Guodong Zhao, and Zhi Chen}
\thanks{This work is supported by Scinece and Technology on Information Transmission and Dissemination in Communication Networks Laboratory.}
%\thanks{The corresponding author is Lin Zhang and the contact email is linzhang1913@gmail.com.}}
}
\maketitle

\thispagestyle{empty}

\begin{abstract}
In underlay heterogeneous networks (HetNets), the distance between a macro base station (MBS) and a macro user (MU) is crucial for a small-cell based station (SBS) to control the interference to the MU and achieve the coexistence. To obtain the distance between the MBS and the MU, the SBS needs a backhaul link from the macro system, such that the macro system is able to transmit the information of the distance to the SBS through the backhaul link. However, there may not exist any backhaul link from the macro system to the SBS in practical situations. Thus, it is challenging for the SBS to obtain the distance. To deal with this issue, we propose a median based (MB) estimator for the SBS to obtain the distance between the MBS and the MU without any backhaul link. Numerical results show that the estimation error of the MB estimator can be as small as $4\%$.
\end{abstract}

%\begin{keywords}
%Cognitive radio, Channel Gain, Maximum likelihood, Median, Estimation accuracy.
%
%\end{keywords}

%\pagestyle{empty}
%\newpage
\section{Introduction}

The \emph{heterogeneous network} (HetNet) is a promising candidate to provide flexible wireless accesses for future wireless communications \cite{Survey_Two_tire_1}. Within a HetNet, a \emph{macro base station} (MBS) is expected to provide wide coverage of users in the macro cell. Meanwhile, a \emph{small-cell base station} (SBS) in the macro cell coexists with the MBS and is responsible for providing high data rate services for the users in the small cell. An effective way to achieve the coexistence between the small cell and the macro cell is underlay HetNet \cite{Survey_Two_tire_2}, \cite{Survey_Two_tire_3}. In the underlay HetNet, the SBS is allowed to access the macro frequency band to enhance the small-cell throughput, provided that the interference from the SBS to a \emph{macro user} (MU) is well controlled.

To implement an underlay HetNet, the location information of the MU is important for the SBS to control its transmit power and manage the interference to the MU. Intuitively, if the MU is out of the SBS's coverage, the SBS is allowed to maximize the transmit power and achieve a high small-cell throughput. Otherwise, if the MU is in the SBS's coverage, the SBS has to carefully control the transmit power to avoid severe interference to the MU. To obtain the location information of the MU, the SBS needs a backhaul link from the macro system, such that the macro system is able to transmit the location information of the MU to the SBS through the backhaul link \cite{Backhaul_1, Backhaul_2, Backhaul_3}. However, there may not exist any backhaul link from the macro system to the SBS in practical situations. Thus, it is challenging for the SBS to obtain the location information of the MU for the underlay HetNet.

Recently, the underlay HetNet has been extensively studied \cite{Two_tire_1, Two_tire_2, Two_tire_3}. Instead of utilizing the instantaneous location information of the MU to manage the interference, these literature exploited the stochastic geometry information of the MU and enabled the SBS to satisfy an average interference constraint. In particular, the \emph{Poisson point process} (PPP) model is widely adopted to enhance the access probability of the SBS, while guaranteeing a small outage probability of the MU. However, this approach can only provide limited access probability of the SBS and compromises the small-cell throughput.

To deal with this issue, we intend to estimate the distance between the MBS and the MU, such that the SBS can exploit the distance to achieve the underlay HetNet. Briefly, if the distance is too small or too large compared with the the distance between the MBS and the SBS, the SBS's transmission will not interfere with the MU since the MU is almost surely beyond the SBS's coverage in this case. If the distance between the MBS and the MU is comparable with the distance between the MBS and the SBS, the SBS has to carefully control its transmit power to avoid severe interference to the MU since it is likely that the MU is in the SBS's coverage.

In fact, it is possible for the SBS to obtain the distance between the MBS and the MU without any backhaul link from the macro system. In principle, if the MBS is transmitting data to the MU with a target SNR, the transmitted signal is designed based on the distance between the MBS and the MU. In particular, if the distance is small, the MBS is able to satisfy the target SNR with a small transmit power. Otherwise, the MBS increases its transmit power to achieve the target SNR. In other words, the transmitted signal from the MBS contains some information of the distance between the MBS and the MU. Thus, the SBS can estimate the distance by sensing the transmitted signal from the MBS.

In this paper, we first study the relation between the MBS signal and the distance from the MBS to the MU. With this relation, we enable the SBS to sense the MBS signal and develop a \emph{median based} (MB) estimator to obtain the distance. Numerical results show that the estimation error of the MB estimator can be as small as $4\%$.

\section{System Model}
Fig. {\ref{systemmodel}} provides a HetNet model, which consists of a MBS, a MU, and a SBS. In particular, the MBS is transmitting data to the MU. Meanwhile, the SBS intends to estimate the distance $d_0$ between the MBS and the MU for the underlay HetNet. In what follows, we present the channel model and the signal model, respectively.

         \begin{figure}[t!]
            \centering
            \includegraphics[scale=0.55]{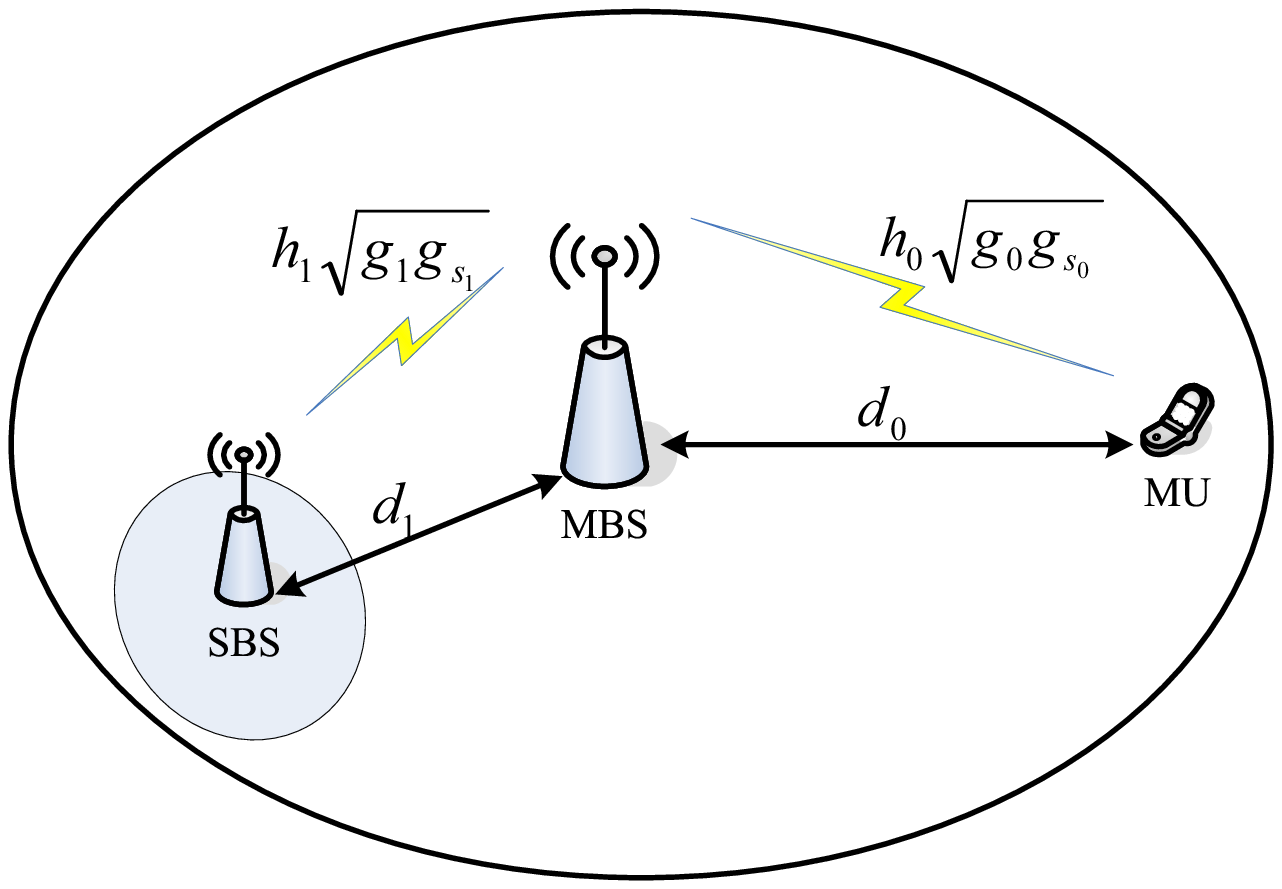}
            \caption{HetNet model, where a MBS is transmitting data to the MU in a HetNet. Meanwhile, the SBS intends to estimate the distance $d_0$ between the MBS and the MU to achieve the underlay HetNet. }
            \label{systemmodel}
        \end{figure}

\subsection{Channel Model}

We denote $h_0$ ($h_1$), $g_{s_0}$ ($g_{s_1}$), and $g_0$ ($g_1$) as the fading, the shadowing, and the path-loss coefficients between the MBS and the MU (SBS), respectively. Then, the channel between the MBS and the MU (SBS) is ${h_0}\sqrt {{g_0}{g_{s_0}}}$ (${h_1}\sqrt {{g_1}{g_{s_1}}}$). In particular, $|h_q|$ ($q=0, \ 1$) follows a Rayleigh distribution with unit mean. $g_{s_q}$ ($q=0, \ 1$) follows a log-normal distribution with variance $\sigma_s^2$. $g_q$ ($q=0, \ 1$) is determined by the path-loss model. If we adopt the path-loss model \cite{3GPP_channel_model}
\begin{equation}
    P_l(d_q)=128 + 37.6\log_{10}(d_q), \ \ \text{for} \ \ d_q \geq
    0.035 \ \text{km},
    \label{Path_loss_model}
\end{equation}
where $d_q$ is the distance between two transceivers, $g_q$ can be expressed as
\begin{equation}
    g_q=10^{-12.8} d_q^{-3.76}, \ \ \text{for} \ \ d_q \geq 0.035 \ \text{km}.
    \label{g_i}
\end{equation}

For illustrations, we provide the channel model in Fig. \ref{Channel_model}, where time axis is divided into blocks and each block consists of multiple subblocks. In particular, $g_q$ remains constant all the time with a given $d_q$, $g_{s_q}$ ($q=0, \ 1$) remains constant within each block ($i$) and varies independently among different blocks, and $h_q$ ($q=0, \ 1$) remains constant within each subblock ($i, j$) and varies among different subblocks.

         \begin{figure}[t!]
            \centering
            \includegraphics[scale=0.45]{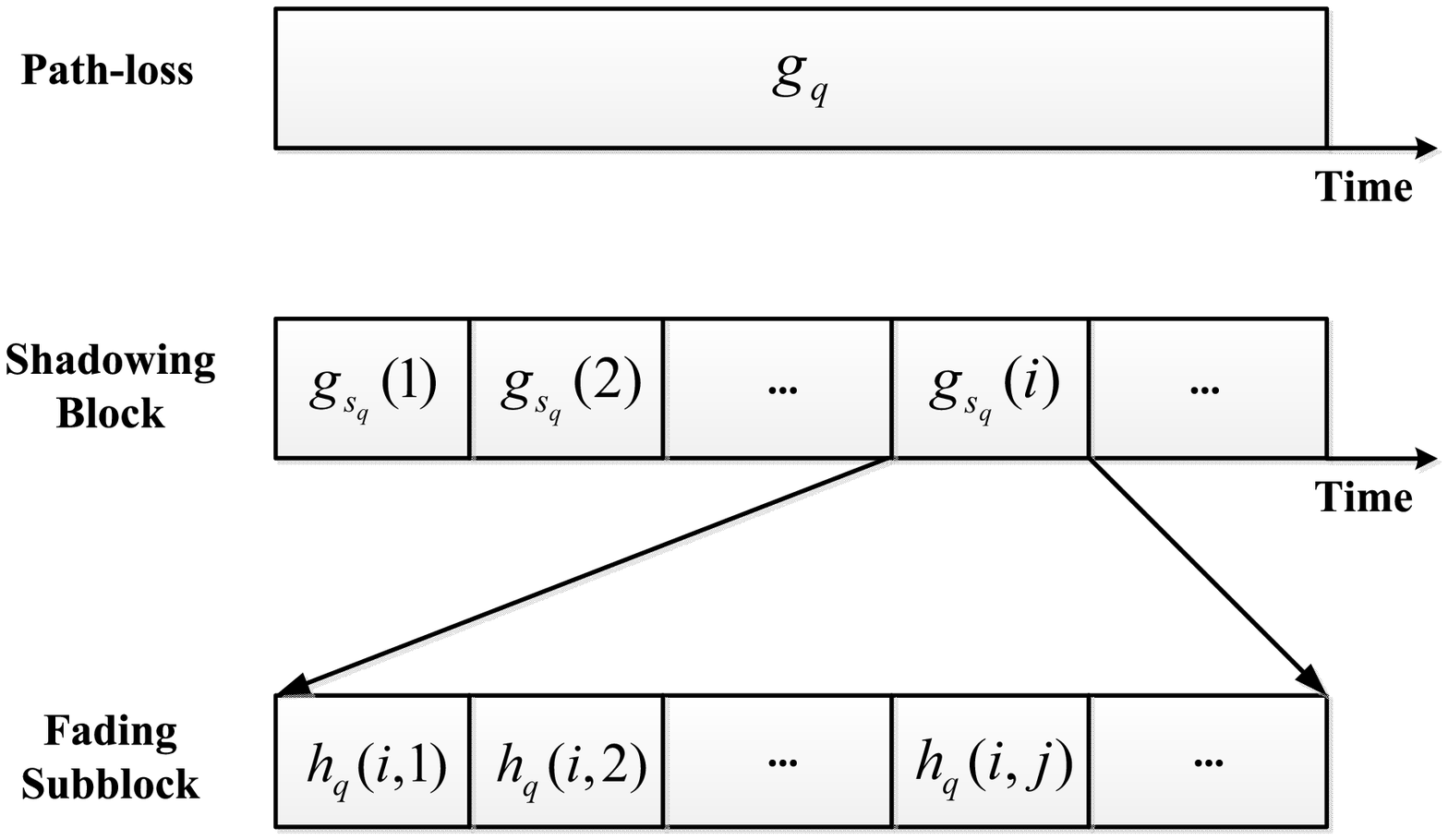}
            \caption{Channel model, where time axis is divided into blocks and each blocks consists of multiple subblocks. In particular, $g_q$ remains constant all the time with a given $d_q$, $g_{s_q}$ ($q=0, \ 1$) remains constant within each block ($i$) and varies independently among different blocks, and $h_q$ ($q=0, \ 1$) remains constant within each subblock ($i, j$) and varies among different subblocks.}
            \label{Channel_model}
        \end{figure}
\subsection{Signal Model}

\subsubsection{Signal model from the MBS to the MU}
Denote $x_{m}$ as the MBS signal with unit power, i.e., $|x_{0}|^2=1$. If the MBS transmits the signal with power $p_0$, the received signal at the MU is
 \begin{equation}
 y_0(i,j)= h_0(i,j)\sqrt {g_0g_{s_0}(i)p_0(i,j)} x_0(i,j) + n_0(i,j),
 \label{y_p}
 \end{equation}
where $(i)$ denotes the index of the $i$th block, $(i,j)$ represents the index of the $j$th subblock in the $i$th block, $n_0$ represents the AWGN at the MU with zero mean and variance $\sigma^2$. Then, the SNR of the received signal at the MU is
 \begin{equation}
\gamma_{0}(i,j)= \frac{|h_0(i,j)|^2g_0g_{s_0}(i)p_0(i,j)}{\sigma^2}.
 \label{gamma_p}
 \end{equation}

We further consider that the MBS and the MU adopt \emph{close loop
power control} (CLPC) to provide QoS guaranteed wireless communication \cite{Rui}. That is, the MBS automatically adjusts its transmit power to meet a certain target SNR $\gamma_T$ at the MU. Then, MBS's transmit power is
\begin{equation}
p_0(i,j)= \frac{{\gamma _T\sigma^2}}{{|h_0(i,j)|^2{g_0g_{s_0}(i)}}}.
\label{p_0}
\end{equation}

\subsubsection{Signal model from the MBS to the SBS}
In the meantime, the received MBS signal at the SBS is
  \begin{equation}
y_1(i,j) = h_1(i,j)\sqrt {{g_1}{g_{s_1}(i)}{p_0(i,j)}} {x_{0}}(i,j) + n_1(i,j),
 \label{y_c}
 \end{equation}
where $n_1$ is the AWGN at the SBS with zero mean and variance $\sigma^2$. Then, the SNR of the received MBS signal at the SBS is
 \begin{equation}
\gamma_1(i,j) = \frac{{|h_1(i,j)|^2{g_1}{g_{s_1}(i)}{p_{0}(i,j)}}}{{\sigma^2}}.
 \label{gamma_c}
 \end{equation}

By Substituting (\ref{g_i}) and (\ref{p_0}) into (\ref{gamma_c}), $\gamma_1(i,j)$ in (\ref{gamma_c}) can be rewritten as
\begin{equation}
{\gamma _{1}}(i,j) = \frac{{{\gamma _T}{d_1^{-3.76}}}}{{{d_0^{-3.76}}}}\frac{g_{s_1}(i)}{g_{s_0}(i)} \frac{{|h_1(i,j)|^2}}{{|h_0(i,j)|^2}}.
\label{gamma_c_1}
\end{equation}

\section{Median Based (MB) Estimator}
In this section, we will develop a MB estimator to obtain the distance $d_0$ between the MBS and the MU. In what follows, we provide the basic principle of the estimator followed by the estimator design.

\subsection{Basic Principle}

From (\ref{gamma_c_1}), each SNR of the received MBS signal at the SBS is highly related to the distance $d_0$. Then, it is possible for the SBS to measure the SNR of the received MBS signal and estimate $d_0$. However, it is difficult to obtain $d_0$ directly. This is because, each SNR is also affected by random Rayleigh fadings and shadowing attenuations, and varies independently among different blocks and/or subblocks. Alternatively, the SBS can measure different SNRs of MBS signals in multiple blocks and utilize the distribution knowledge of the Rayleigh fadings and shadowing attenuations to estimate $d_0$.

Next, we will first calculate the \emph{cumulative density function} (CDF) of the SNR at the SBS. Then, we study the relation between the CDF of the SNR at the SBS and the distance $d_0$. Finally, we develop the MB estimator.

\subsection{CDF of the SNR at the SBS}
Removing the time index of the SNR at the SBS in (\ref{gamma_c_1}) and rewriting the SNR in dB, we have
\begin{align}
\nonumber
{\gamma _{1,dB}} = {\gamma _{T,dB}} + 37.6{\log _{10}}\left( \frac{d_0}{d_1} \right)+ \Theta _r+\Theta _s,
 \label{gamma_c_dB}
\end{align}
where ${\Theta _r} = 10{\log _{10}}\left( {\frac{{h_1^2}}{{h_0^2}}} \right)$ and ${\Theta _s} = 10{\log _{10}}\left( {{g_{s_1}}} \right) - 10{\log _{10}}\left( {{g_{s_0}}} \right)$. Next, we first calculate the \emph{probability density function} (PDF) of $\Theta _r$ and $\Theta _s$, and then obtain the CDF of $\gamma _{1,dB}$.

On one hand, both $|h_0|$ amd $|h_1|$ follow a Rayleigh distribution with unit mean, the CDF of $\phi=|h_1|^2/|h_0|^2$ is \cite{Joint_PDF}
\begin{equation}
F_{\Phi}(\phi) = \frac{{\phi}}{{{{1 + \phi}}}}.
\label{phi_cdf}
\end{equation}

Then, the CDF of $\Theta _r=10\log_{10}(\phi)$ is
\begin{align} \nonumber
F_{\Theta _r}(\theta _r) =& \Pr\left\{10\log_{10}(\phi) \leq \theta _r \right\}\\ \nonumber
=& \Pr\left\{\phi \leq 10^{\frac{\theta _r}{10}} \right\}\\
=& F_{\Phi}\left( 10^{\frac{\theta _r}{10}}\right).
\label{Theta_r_cdf}
\end{align}

Substituting (\ref{phi_cdf}) into (\ref{Theta_r_cdf}), we have the CDF of $\Theta _r$ as
\begin{equation}
F_{\Theta _r}(\theta _r) = \frac{1}{1 + 10^{-\frac{\theta _r}{10}}}.
\label{Theta_1_cdf}
\end{equation}

Taking the derivation of $F_{\Theta _r}(\theta _r)$, we have the PDF of $\Theta _r$ as
\begin{equation}\label{a5}
{f_{{\Theta _r}}}\left( {{\theta _r}} \right) = \frac{{\ln 10 \cdot {{10}^{ - \frac{{{\theta _r}}}{{10}}}}}}{{10{{\left( {1 + {{10}^{ - \frac{{{\theta _r}}}{{10}}}}} \right)}^2}}}.
\end{equation}

On the other hand, both $g_{s_0}$ and $g_{s_1}$ follow a log-normal distribution with variance $\sigma_s^2$. Then, it is straightforward to obtain that $\Theta_s$ follows a normal distribution with zero means and variance $2\sigma_s^2$. Thus, the PDF of $\Theta_s$ is
\begin{equation}\label{Theta_2_pdf}
{f_{{\Theta _s}}}\left( {{\theta _s}} \right) = \frac{1}{{\sqrt {4\pi   \sigma _s^2 } }}{e^{ - \frac{{{{ {{\theta _s}} }^2}}}{{4 {\sigma _s^2}}}}}.
\end{equation}

Note that, the CDF of $\gamma_{1, dB}$ is
\begin{align}
\label{gamma_c_dB_CDF_1}
\nonumber
& F_{\Gamma _{1,dB}}(\gamma_{1, dB}) \\ \nonumber
= &\Pr\left\{{\gamma _{T,dB}} + 37.6{\log _{10}}\left( \frac{d_0}{d_1} \right)+ {\Theta _r} + {\Theta _s}\leq \gamma_{1, dB}\right\}\\
= &\Pr\left\{{\Theta _r}\! + {\Theta _s}\leq \gamma_{1, dB}\!-{\gamma _{T,dB}}\!- 37.6{\log _{10}}\left( \frac{d_0}{d_1} \right)\right\}.
\end{align}

If we denote $m(\gamma_{1, dB})=\gamma_{1, dB}-{\gamma _{T,dB}}- 37.6{\log _{10}}\left( \frac{d_0}{d_1} \right)$, the CDF of $\gamma_{1, dB}$ in (\ref{gamma_c_dB_CDF_1}) can be calculated as
\begin{align}
\label{gamma_c_dB_CDF_2}
\nonumber
 F_{\Gamma _{1,dB}}(\gamma_{1, dB})=& \Pr\left\{{\Theta _r} + {\Theta _s}\leq m(\gamma_{1, dB})\right\}\\ \nonumber
=&\int_{ - \infty }^\infty  \int_{ - \infty }^{m(\gamma_{1, dB}) - {\theta _s}}\!\!f_{\Theta _s}\left(\theta _s \right) \!\!f_{\Theta _r}\left( \theta _r \right)d\theta _rd\theta _s\\
\!=\!&\int_{ - \infty }^\infty \!\! f_{\Theta _s}\left(\theta _s \right) F_{\Theta _r}\left(m(\gamma_{1, dB})\! - \! {\theta _s} \!\right)d\theta _s.
\end{align}

\subsection{Relations between the CDF of $\gamma_{1, dB}$ and the Distance $d_0$}
To begin with, we provide the definition of the median $x_{\frac{1}{2}}$ of a random variable $X$ as follows,

\begin{definition} For a random variable $X$ with CDF $F_X(x)$, $x\in \mathbb{R}$, if $x_{\frac{1}{2}}$ satisfies both $F_X(x_{\frac{1}{2}})= \Pr\{X\leq x_{\frac{1}{2}}\}=\frac{1}{2}$ and $1-F_X(x_{\frac{1}{2}})= \Pr\{X\geq x_{\frac{1}{2}}\}=\frac{1}{2}$, $x_{\frac{1}{2}}$ is defined as the median of the random variable $X$.
\end{definition}

Based on Definition 1, we can obtain the median $\gamma_{1, dB, \frac{1}{2}}$ of the random variable $\gamma_{1, dB}$ by letting $F_{\Gamma_{1,dB}}\left(\gamma_{1,dB}\right)$ in (\ref{gamma_c_dB_CDF_2}) be $\frac{1}{2}$, i.e.,
\begin{align}
\label{median_1}
\int_{ - \infty }^\infty  f_{\Theta _s}\left(\theta _s \right) F_{\Theta _r}\left(m(\gamma_{1, dB,\frac{1}{2}}) - {\theta _s} \right)d\theta _s=\frac{1}{2}.
\end{align}

Then, we have the following Theorem.

\begin{theorem}
The median $\gamma_{1, dB, \frac{1}{2}}$ enables $m(\gamma_{1, dB,\frac{1}{2}})=0$ to hold.
\end{theorem}
\begin{proof}
The detailed proof of this Theorem is provided in the Appendix.
\end{proof}

Theorem 1 indicates that the relation between the median $\gamma_{1, dB, \frac{1}{2}}$ and the distance $d_0$ satisfies
\begin{align}
\label{median_2}
m(\gamma_{1, dB,\frac{1}{2}})\!=\!\gamma_{1, dB, \frac{1}{2}}\!-\!{\gamma _{T,dB}}\!- \!37.6{\log _{10}}\left( \frac{d_0}{d_1} \right)=0.
\end{align}
Thus, if $\gamma_{1, dB, \frac{1}{2}}$ is available at the SBS, $d_0$ can be directly calculated with (\ref{median_2}). However, $\gamma_{1, dB, \frac{1}{2}}$ is unknown to the SBS in practical situations. To deal with this issue, we will first estimate $\gamma_{1, dB, \frac{1}{2}}$ and then obtain the estimation of $d_0$ with (\ref{median_2}).

\subsection{Estimator Design}

We first give the definition of the sample median $x^s_{\frac{1}{2}}$ of a random variable $X$ as follows,

\begin{definition} For a random variable $X$ with samples $x_m$ ($1\leq m \leq M$), if $x^{s}_{\frac{1}{2}}$ satisfies both $\Pr\{x_m\leq x^{s}_{\frac{1}{2}}\}=\frac{1}{2}$ and $\Pr\{x_m \geq x^{s}_{\frac{1}{2}}\}=\frac{1}{2}$, $x^{s}_{\frac{1}{2}}$ is defined as the sample median of the random variable $X$.
\end{definition}

If the SBS observes MBS signals in $I$ blocks and measures $\gamma_{1, dB}$ of $J$ subblocks within each block, the SBS is able to measure $K=IJ$ independent samples of $\gamma_{1, dB}$, namely, $\gamma _{1, dB}(i,j)$ ($1\leq i\leq I$, $1\leq j\leq J$). In what follows, we will approximate the median $\gamma_{1, dB, \frac{1}{2}}$ with the sample median $\gamma^s_{1, dB, \frac{1}{2}}$ of these $K$ samples. Then, by substituting the approximated $\gamma_{1, dB, \frac{1}{2}}$ into (\ref{median_2}), we obtain the estimation of $d_0$.

To begin with, by sorting the $K$ samples in ascending order, the $K$ samples can be relabelled as $\bar{\gamma}_{1, dB}(k)$ ($1\leq k\leq K$), i.e., $\bar{\gamma} _{1, dB}(k_1) \leq \bar{\gamma} _{1, dB}(k_2)$ for $1\leq k_1 \leq k_2 \leq K$. Since the sample medians $\bar \gamma^s_{1, dB, \frac{1}{2}}$ of these $K$ samples for odd and even $K$ can be different, we will develop the MB estimator for odd and even $K$ separately.

 \subsubsection{For the case that $K$ is odd} When $K$ is odd, the sample median is $\gamma^s_{1, dB, \frac{1}{2}}=\bar \gamma _{1, dB}\left(\frac{K+1}{2}\right)$. Then, the median of $\gamma_{1, dB}$ can be approximated as
\begin{equation}
\gamma_{1, dB, \frac{1}{2}}\approx \bar \gamma _{1, dB}\left(\frac{K+1}{2}\right).
\label{approximate_odd}
\end{equation}

By substituting (\ref{approximate_odd}) into (\ref{median_2}), we have the MB estimator as
\begin{equation}
\begin{split}
{{\hat d}_{0}} = d_110^{\frac{\bar \gamma _{1, dB}\left(\frac{K+1}{2}\right)-\gamma _{T,dB}}{37.6}}.
\end{split}
\end{equation}

 \subsubsection{For the case that $K$ is even} When $K$ is even, the sample median is between $\bar \gamma_{1, dB}\left(\frac{K}{2}\right)$ and $\bar \gamma _{1, dB}\left(\frac{K}{2}+1\right)$. Then, the median of $\gamma_{1, dB}$ can be approximated as
\begin{equation}
\gamma_{1, dB, \frac{1}{2}}\approx \frac{\bar \gamma_{1, dB}\left(\frac{K}{2}\right)+\bar \gamma _{1, dB}\left(\frac{K}{2}+1\right)}{2}.
\label{approximate_even}
\end{equation}

By substituting (\ref{approximate_even}) into (\ref{median_2}), we have the MB estimator as
\begin{equation}
\begin{split}
{{\hat d}_{0}} = d_110^{\frac{\frac{\bar \gamma_{1, dB}\left(\frac{K}{2}\right)+\bar \gamma _{1, dB}\left(\frac{K}{2}+1\right)}{2}-\gamma _{T,dB}}{37.6}}.
\label{MB_estimator_even}
\end{split}
\end{equation}

Consequently, the MB estimator can be summarized as
\begin{equation}\label{MB_estimator}
{{\hat d}_{0}} \!=\! \left\{\!\! \begin{array}{l}
d_110^{\frac{\bar \gamma _{1, dB}\left(\frac{K+1}{2}\right)-\gamma _{T,dB}}{37.6}}, \quad \quad \quad \quad \ \ \text{for} \ K \ \text{is}\ \text{odd,}\\
d_110^{\frac{\frac{\bar \gamma_{1, dB}\left(\frac{K}{2}\right)+\bar \gamma _{1, dB}\left(\frac{K}{2}+1\right)}{2}-\gamma _{T,dB}}{37.6}}, \ \text{for} \ K \ \text{is} \ \text{even}.
\end{array} \right.
\end{equation}
From (\ref{MB_estimator}), the MB estimator $\hat{d}_{0}$ is determined by the target SNR $\gamma_{T, dB}$ at the MU, the distance $d_1$ between the MBS and the SBS, and the SNR $\gamma_{1, dB}$ of the MBS signal at the SBS. Note that, $\gamma_{T, dB}$ can be obtained by the SBS through observing the \emph{modulation and coding scheme} (MCS) of the MBS signal \cite{D_Tse}. $d_1$ is available at the SBS. $\gamma_{1, dB}$ is measured at the SBS and also known to the SBS. Therefore, the estimation of $d_0$ can be directly calculated with (\ref{MB_estimator}). In other words, the computational complexity of the MB estimator in (\ref{MB_estimator}) is $O(1)$.

\subsection{Estimation Performance Analysis}

In this part, we present the estimation performance analysis of the MB estimator in Theorem 1.

\begin{theorem} With the MB estimator in (\ref{MB_estimator}), $g_0$ can be bounded by
\begin{equation}
\!d_110^{\frac{\gamma _{c, dB}\left(\left\lceil \frac{IJ+1}{2}\right\rceil-1\right)-\gamma _{T,dB}}{37.6}} \!\!\leq \!d_0 \!\leq \!\!d_110^{\frac{\gamma _{c, dB}\left(\left\lfloor \frac{IJ+1}{2}\right\rfloor+1\right)-\gamma _{T,dB}}{37.6}}\!
\label{Theorem_2}
\end{equation}
with probability $\left(1-\left( {\frac{1}{2}} \right)^{\left[\frac{K}{2} + 1\right]}\right)^2$, where $\lceil x\rceil$ denotes the smallest integer that is no smaller than $x$, $\lfloor x \rfloor$ denotes the largest integer that is no larger than $x$, and $[x]$ rounds $x$ to the nearest integer.
\end{theorem}

%\begin{proof}
%The detailed proof of this Theorem is provided in Appendix A.
%\end{proof}

Theorem 2 indicates that $d_0$ can be upper bounded and lower bounded by functions of the measured SNRs at the SBS with a certain probability. In particular, the probability is a function of the number of the sensed primary signals, i.e., $IJ$. For instance, we consider $I=12$ and $J=1$, the bounds of $d_0$ in (\ref{Theorem_2}) holds with probability larger than $98\%$. As $IJ$ increases, $d_0$ can be almost surely bounded with (\ref{Theorem_2}).
             \begin{figure}[t!]
            \centering
            \includegraphics[scale=0.4]{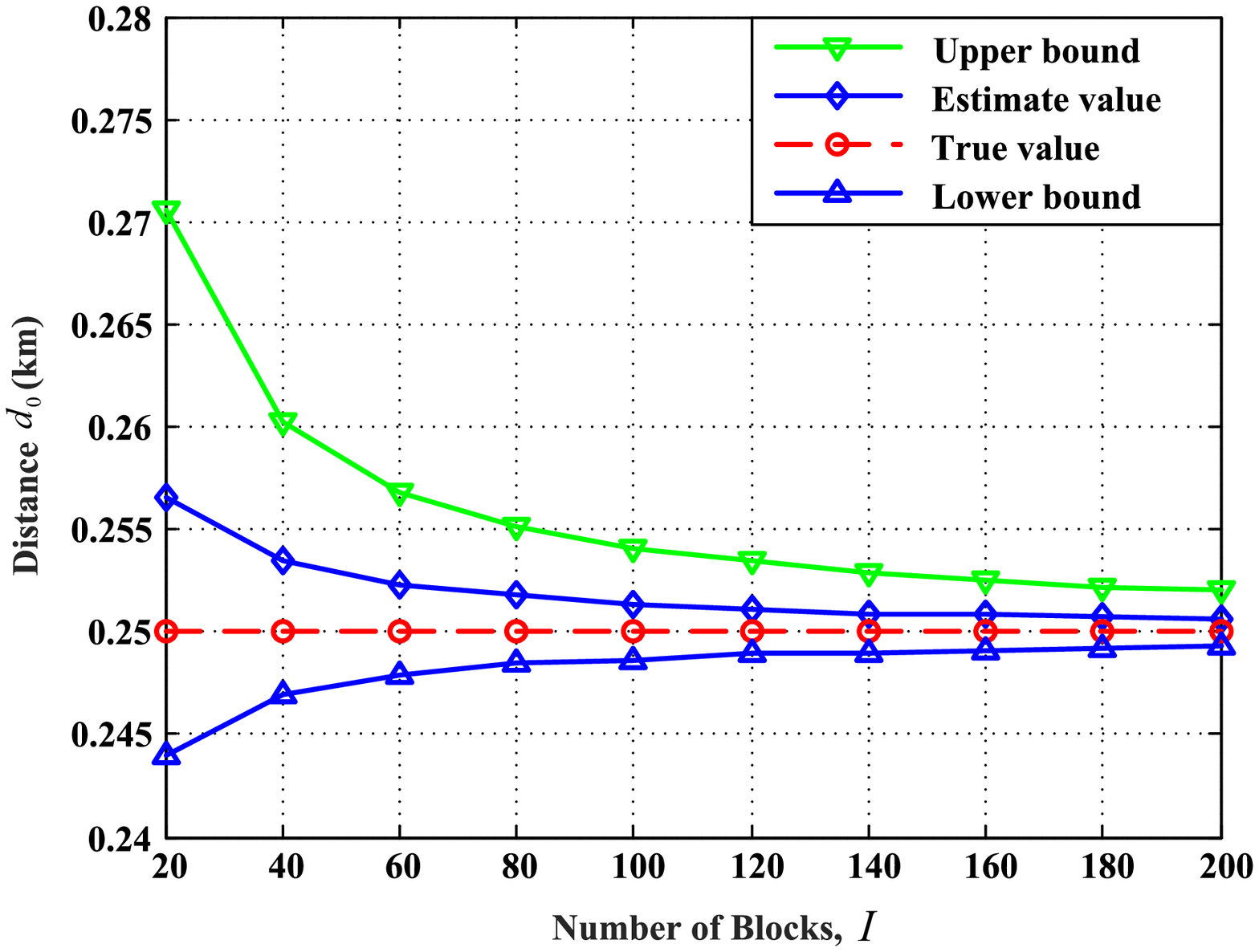}
            \caption{Estimation accuracy of the MB estimator versus the number of blocks, i.e., $I$. In particular, we set $J=1$.}
            \label{bound1}
        \end{figure}

In Fig. \ref{bound1}, we illustrate estimation accuracy of the MB estimator versus the number of blocks, i.e., $I$. In particular, we provide the true value of the channel gain $d_0$, the estimation $\hat{d}_0$ with the MB estimation, the upper bound and the lower bound in (\ref{Theorem_2}). Here, the distance between the MBS and the MU is $d_0=0.25$ km, and the distance between the MBS and the MU is $d_1=0.1$ km. From this figure, we observe that the estimation value $\hat{d}_0$ is strictly upper bounded and lower bounded by the results in Theorem 1. Besides, both the upper bound and the lower bound converge to $d_0$ and $\hat{d}_0$. This means a larger $I$ leads to a more accurate estimation as well as tighter upper and lower bounds. Furthermore, we observe that the estimated value $\hat{d}_0$ converges to the true value of the channel gain $d_0$ as $I$ increases.

\section{Numerical results}
In this section, we provide the numerical results to demonstrate the performance of the proposed MB estimator. Here, we adopt the system model as in Section II, where the radius of the MBS's coverage is $R=0.5$ km, the power of the AWGN $\sigma^2=-114$ dBm, the target SNR of the MU is $\gamma_{T}=10$ dB. Furthermore, $10^4$ Monte Carlo trails are conducted for each curve.
To begin with, we define the estimation error of $d_0$ as $\epsilon = |\frac{\hat d_0-d_0}{d_0}|$.

             \begin{figure}[t!]
            \centering
            \includegraphics[scale=0.43]{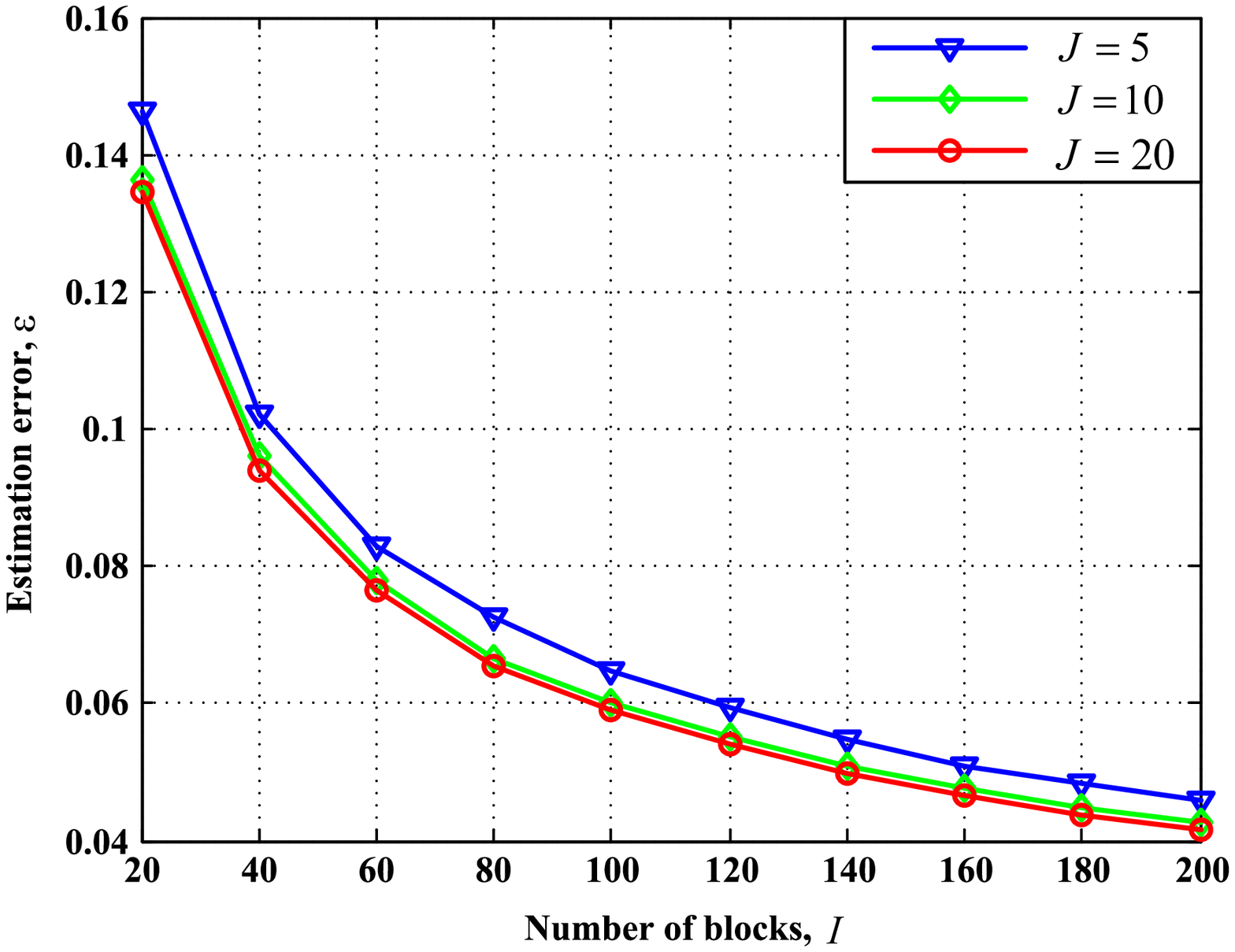}
            \caption{The estimation error with different $I$ and/or $J$. In particular, we set $d_0=0.25$ km and $d_1=0.1$ km.}
            \label{error_I}
        \end{figure}

Fig. \ref{error_I} investigates the impact of $K=IJ$ on the the estimation error. In particular, we set $d_0=0.25$ km and $d_1=0.1$ km. From this figure, the estimation error decreases as $J$ ($I$) grows. This is reasonable, a larger $J$ ($I$) leads to more SNR samples at the SBS and provides more accurate information of the median of the SNRs at the SBS. By using the relation between the median of the SNRs at the SBS and the distance $d_0$, the MB estimator is able to output a more accurate estimation of $d_0$. Besides, we observe that a small increase of $I$ leads to big jump of the estimation error, while a large increase of $I$ results in a slight decrease of the estimation error. This indicates that the estimation error is more sensitive to $I$ than $J$. Thus, It is an effective way to achieve a small estimation error by adopting a large $I$ and a reasonable $J$.

             \begin{figure}[t!]
            \centering
            \includegraphics[scale=0.43]{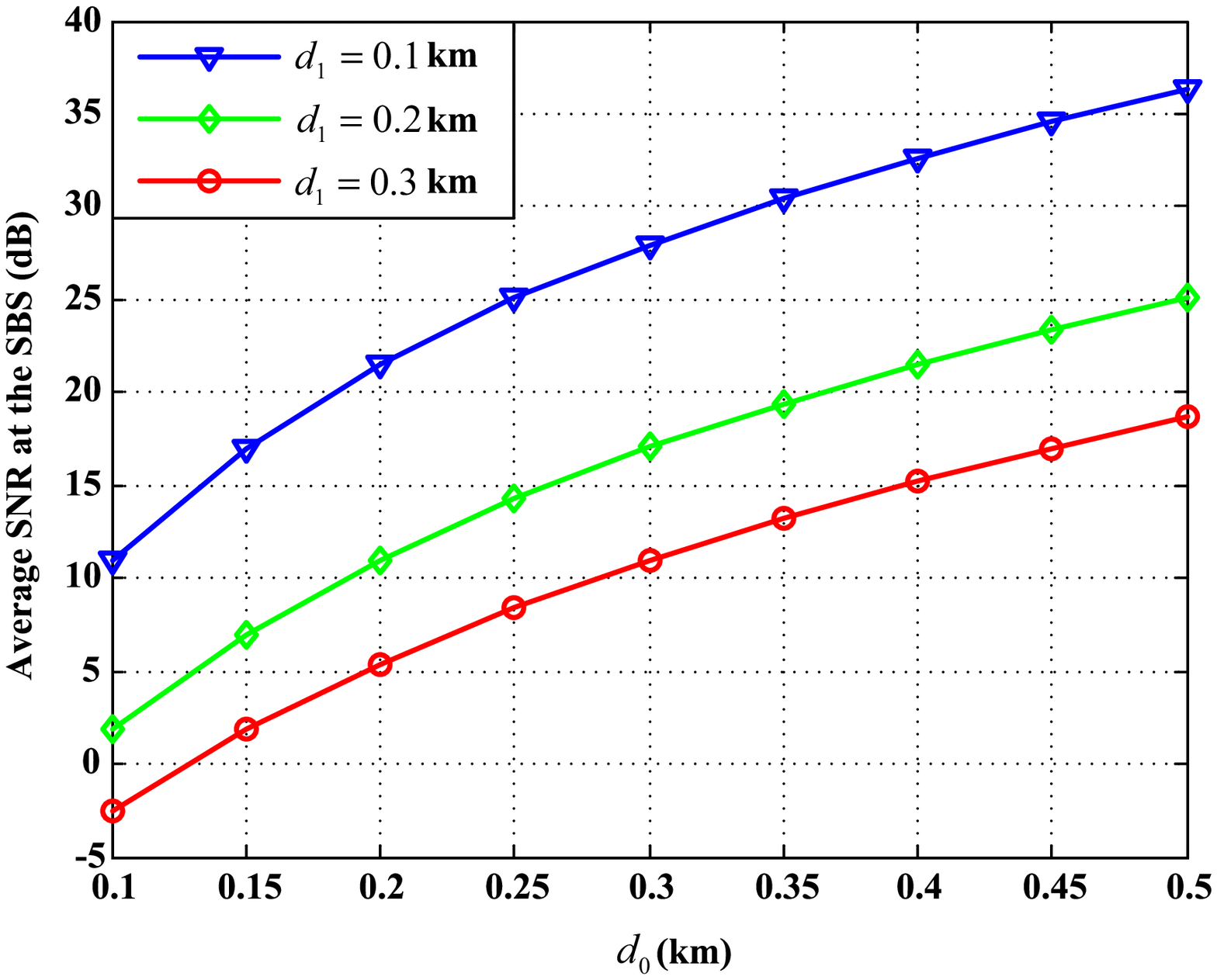}
            \caption{The average measured SNR at the SBS when the MU and/or the SBS are in different locations, i.e., different $d_0$ and/or $d_1$. In particular, we set $I=200$ and $J=20$.}
            \label{SNR}
        \end{figure}

Fig. \ref{SNR} provides the average measured SNRs at the SBS when the MU and/or the SBS are in different locations, i.e., different $d_0$ and/or $d_1$. In general, the SNR at the SBS increases as $d_0$ grows or $d_1$ decreases. On one hand, for a given target SNR of the received signals at the MU, a larger $d_0$ requires a larger transmit power at the MBS to satisfy the target SNR. This leads to a stronger received signal at the SBS and outputs a larger SNR. On the other hand, a smaller $d_1$ also enables the SBS to receive a stronger signal from the MBS and contributes to a larger SNR at the SBS.

             \begin{figure}[t!]
            \centering
            \includegraphics[scale=0.43]{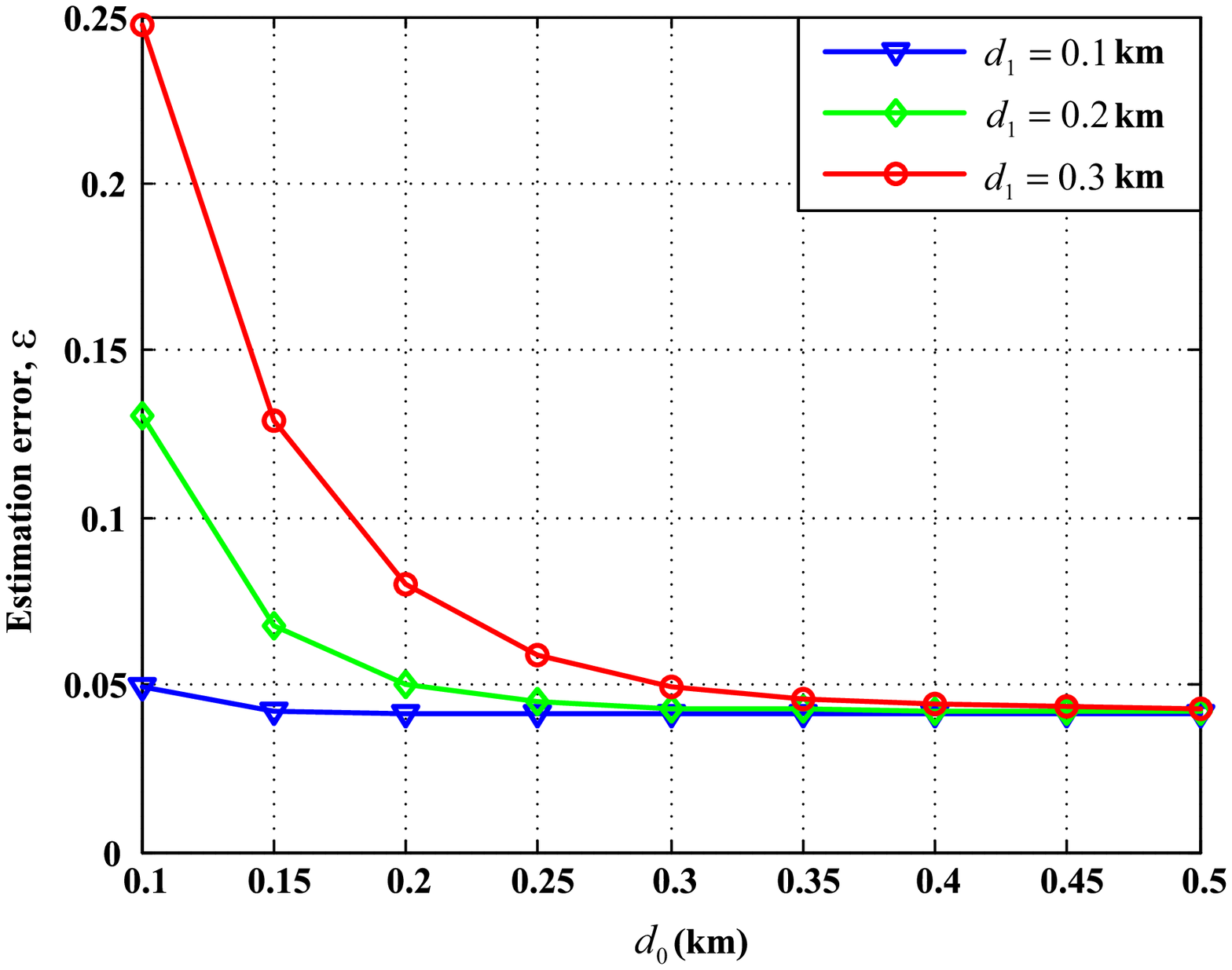}
            \caption{The estimation error when the MU and/or the SBS are in different locations, i.e., different $d_0$ and/or $d_1$. In particular, we set $I=200$ and $J=20$.}
            \label{error_d_0}
        \end{figure}

Fig. \ref{error_d_0} provides the estimation error when the MU and/or the SBS are in different locations, i.e., different $d_0$ and/or $d_1$. From this figure, the estimation error decreases as $d_0$ increases or $d_1$ decreases. Note that, both the increase of $d_0$ and the decrease of $d_1$ enhances the average SNR at the SBS from the results in Fig. \ref{SNR}. This reduces the measure error of each SNR at the SBS. By adopting these SNRs to estimate the distance $d_0$, the estimation error is also reduced. Besides, we observe that the estimation error converges to around $4\%$ as $d_0$ increases or $d_1$ decreases. In addition, by comparing Fig. \ref{SNR} and Fig. \ref{error_d_0}, we have the estimation error of $d_0$ with different average SNRs at the SBS as in Table I. In fact, the estimation error in Table I can be further reduced by increasing $I$ and $J$ from the results in Fig. \ref{error_I}. Therefore, we can select a proper $I$ and $J$ to obtain an acceptable estimation performance in practical situations.

\begin{table}[!hbp]
\caption{Estimation error of $d_0$ with different average SNRs at the SBS. }
\centering %\footnotesize
\begin{tabular}{|c|c|c|c|c|c|c|}
  \hline
  Average $\gamma_{1, dB}$ & 0 & 5 & 10 & 15 & 20 & $\cdots$ \\
  \hline
  $\epsilon$ & 16\% & 8\% & 5\% & 4.5\% & 4\% & 4\% \\
  \hline
\end{tabular}
\end{table}

%This is because, as the average SNR at the SBS increases, the measure error of each SNR decreases. When the measure error of each SNR is small enough, the estimation error of $d_0$ caused by the measure error of each SNR can be neglected. Then, the estimation error of $d_0$ is only determined by the number of SNR samples at the SBS, i.e., $IJ$. Thus, to further reduce the estimation error of $d_0$, $I$ or $J$ has to be increased.

\section{Conclusions}
In this paper, we studied the coexistence problem between a macro cell and a small cell in an underlay HetNet. In particular, we proposed a MB estimator for the SBS to estimate the distance between the MBS and the MU. Different from the conventional approach, the MB estimator does not require any backhaul link from the macro system to the SBS. With the distance information, the SBS is able to manage the interference to the MU and achieve the coexistence. Numerical results showed that the estimation error of the MB estimator can be as small as $4\%$.

\section{Appendix}

To prove Theorem 1, we only need to verify $\int_{ - \infty }^\infty  f_{\Theta _s}\left(\theta _s \right) F_{\Theta _r}\left(- {\theta _s} \right)d\theta _s=\frac{1}{2}$. Substituting (\ref{Theta_1_cdf}) into $\int_{ - \infty }^\infty  f_{\Theta _s}\left(\theta _s \right) F_{\Theta _r}\left(- {\theta _s} \right)d\theta _s$, we have
\begin{align}\label{Proof_1}
\nonumber
&\int_{ - \infty }^\infty  f_{\Theta _s}\left(\theta _s \right) F_{\Theta _r}\left(- {\theta _s} \right)d\theta _s\\ \nonumber
=&\!\int_{ - \infty }^0  \!\!\!f_{\Theta _s}\!\! \left(\theta _s \right)\frac{1}{1 \!+\! 10^{\frac{\theta _s}{10}}}d\theta _s\!\!+\!\!\int_{0 }^\infty \!\!\! f_{\Theta _s}\!\!\left(\theta _s \right)\frac{1}{1\! + \! 10^{\frac{\theta _s}{10}}}d\theta _s\\
\!=&\!\int_{ 0}^\infty  \!\!\! f_{\Theta _s}\!\!\left(-\theta _s \!\right)\frac{1}{1\! + \!10^{-\frac{\theta _s}{10}}}d\theta _s\!\!+\!\!\int_{0 }^\infty \!\!\! f_{\Theta _s}\!\!\left(\theta _s \!\right)\frac{1}{1 \!+ \!10^{\frac{\theta _s}{10}}}d\theta _s.
\end{align}

%=&\int_{ - \infty }^\infty  f_{\Theta _s}\left(\theta _s \right)\frac{1}{1 + 10^{\frac{\theta _s}{10}}}d\theta _s\\ \nonumber

From (\ref{Theta_2_pdf}), we observer that ${{f_{{\Theta _s}}}\left( {{\theta _s}} \right)}$ is an even function. Then, we have $f_{\Theta _s}\left(-\theta _s \right)=f_{\Theta _s}\left(\theta _s \right)$. Meanwhile, we have $\frac{1}{1 + 10^{-\frac{\theta _s}{10}}}=1-\frac{1}{1 + 10^{\frac{\theta _s}{10}}}$. Thus, (\ref{Proof_1}) can be rewritten as
\begin{align}\label{Proof_2}
\nonumber
&\int_{ - \infty }^\infty  f_{\Theta _s}\left(\theta _s \right) F_{\Theta _r}\left(- {\theta _s} \right)d\theta _s\\ \nonumber
=&\int_{ 0}^\infty \!\! \! f_{\Theta _s}\left(\theta _s \right)\left(\!1-\!\frac{1}{1 \!+ \!10^{\frac{\theta _s}{10}}}\!\right)d\theta _s\!+\!\int_{0 }^\infty  \!\!\!f_{\Theta _s}\left(\theta _s \right)\frac{1}{1 \!+ \! 10^{\frac{\theta _s}{10}}}d\theta _s\\ \nonumber
=&\int_{0 }^\infty  f_{\Theta _s}\left(\theta _s \right)d\theta _s\\
=&\frac{1}{2}.
\end{align}
Here, we complete the proof of Theorem 1.

%
%when ${m_{\frac{1}{2}}}=0$, there is
%\begin{equation}\label{a10}
%\Pr \left\{ {{\Gamma _1} + {\Gamma _2} = 0} \right\} = \frac{1}{2}
%\end{equation}
%we let
%\begin{equation}\label{a26}
%y\left( m \right) = \int_{ - \infty }^\infty  {{f_{{\Gamma _2}}}\left( {{\tau _2}} \right) \cdot } \frac{1}{{1 + {{10}^{ - \frac{{m - {\tau _2}}}{{10}}}}}}d{\tau _2}
%\end{equation}
%then,
%\begin{equation}\label{a27}
%y\left( 0 \right) = \int_{ - \infty }^\infty  {{f_{{\Gamma _2}}}\left( {{\tau _2}} \right) \cdot } \frac{1}{{1 + {{10}^{\frac{{{\tau _2}}}{{10}}}}}}d{\tau _2}
%\end{equation}

\end{document}